\documentclass[a4paper,UKenglish]{article}
\usepackage{tikz}
\usepackage[algoruled,linesnumbered,vlined]{algorithm2e}
\usepackage{amsthm}
\usepackage{amsmath}
\usepackage{amsfonts}
\usepackage{url}
\usepackage{hyperref}
\usepackage{geometry}
\usepackage{microtype}
\usepackage{graphicx}

\newtheorem{lemma}{Lemma}
\newtheorem{corollary}{Corollary}
\newtheorem{theorem}{Theorem}

\newtheorem{remark}{Remark}

\title{On Extensions of Maximal Repeats in Compressed Strings}
\author{Julian Pape-Lange\thanks{Technische Universit\"at Chemnitz, Stra\ss e der Nationen 62, 09111 Chemnitz, Germany. Email: julian.pape-lange@informatik.tu-chemnitz.de}}
\date{}

\begin{document}
    
    \maketitle
    
    \begin{abstract}
        This paper provides an upper bound for several subsets of maximal repeats and maximal pairs in compressed strings and also presents a formerly unknown relationship between maximal pairs and the run-length Burrows-Wheeler transform.
        
        This relationship is used to obtain a different proof for the Burrows-Wheeler conjecture which has recently been proven by Kempa and Kociumaka in ``Resolution of the Burrows-Wheeler Transform Conjecture''.
        
        More formally, this paper proves that a string $S$ with $z$ LZ77-factors and without $q$-th powers has at most $73(\log_2 |S|)(z+2)^2$ runs in the run-length Burrows-Wheeler transform and the number of arcs in the compacted directed acyclic word graph of $S$ is bounded from above by $18q(1+\log_q |S|)(z+2)^2$.
    \end{abstract}
    
    \section{Introduction}
    
    A maximal repeat $P$ of a string $S$ is a substring of $S$ which occurs at least twice in $S$ and such that all extensions of $P$ occur less often in $S$. Raffinot proves in \cite{Raffinot:maxRepeatsVsCDAWGs} that there is a natural bijection from the internal nodes in a Compacted Directed Acyclic Word Graph (CDAWG) to the maximal repeats, which is given by the labels of the paths. Also, Furuya et al.\ present in \cite{Furuya_MR-RePair} a relation between maximal repeats and the grammar compression algorithm RePair, and they use this relation to design MR-RePair, an improved variant of RePair.
    
    Sometimes, maximal repeats are not sufficient, since they do not contain any information about the surrounding string. Therefore, in \cite{Belazzougui:lower_bound}, Belazzougui et al.\ introduce the number of right extensions of maximal repeats as a measure for the repetitiveness of strings. They further prove that the number of arcs in the CDAWG is equal to the number of right extensions of maximal repeats and that the number of runs in the run-length Burrows-Wheeler transform (RLBWT) is bounded from above by the number of right extensions of maximal repeats.
    
    In earlier work, I proved in \cite{Pape-Lange} that the number of maximal repeats in a string $S$ with $z$ (self-referential) LZ77-factors and without $q$-th powers is bounded from above by $3q(z+1)^3-2$ and that this upper bound is tight up to a constant factor. This result implies that for a string $S$ over an alphabet $\Sigma$, the number of arcs in the CDAWG, and thereby the number $r$ of runs in the RLBWT, is bounded from above by $3|\Sigma|q(z+1)^3$.
    
    We should expect that of all the $\mathcal{O}\left(qz^3\right)$ maximal repeats some provide less information than others. For example in the string 
    \[ba^{10}ba^{20}b\$=baaaaaaaaaabaaaaaaaaaaaaaaaaaaaab\$\textup{,}\]
    we can derive all maximal pairs of the maximal repeats of $a^i$ from the maximal pairs of $a^9$ and $a^{19}$. In this way, highly-periodic maximal repeats with exponent close to the exponent of the corresponding runs are more important than other maximal repeats which are powers of the same base.
    
    Blumer et al.\ have already shown in 1987 in \cite{Blumer:CDAWG} that the CDAWG cannot compress high powers and that the CDAWG of $a^n\$$ has size $\Theta\left(n\right)$. Contrary to the CDAWG, the RLBWT does not suffer from high powers and we should expect that many right extensions of maximal repeats do not increase the number of runs. And in fact, if the string is very structured, we expect that the output consists of few runs of single characters. For example Christodoulakis et al.\ show in \cite{FibBWT} that the Burrows-Wheeler transform of the $n$-th Fibonacci string $F_n$ is given by $b^{f_{n-2}} a^{f_{n-1}}$.
    
    Yet, until recently, it remained an open question whether there is an upper bound for the number of runs in the RLBWT which is polynomial in the number of LZ77-factors and the logarithm of the length of the string only. This Burrows-Wheeler transform conjecture was resolved in October 2019 by Kempa and Kociumaka who prove in their arXiv-article \cite{Kempa_BWT_Conjecture} that $r\in \mathcal{O}\left(z (\log n)^2\right)$ holds and promised that they will show $r\in \mathcal{O}\left(\delta \log \delta \max \left(\frac {n}{\delta \log \delta}\right)\right)$ for a complexity measure $\delta\leq z$ in an extended version.
    
    This paper provides a different approach to the Burrows-Wheeler transform conjecture and shows by using maximal repeats and their extensions that $r\leq 73(\log_2 |S|)(z+2)^2$ holds.
    
    On the way, this paper also shows that the number of arcs in the CDAWG is bounded from above by $18q(1+\log_q |S|)(z+2)^2$ and gives new insights into the combinatorial properties of extensions of maximal repeats which are either non-highly-periodic or cannot be extended by more than a period length.
    
    \section{Definitions}
    
    Let $\Sigma$ be an \emph{alphabet}.
    A \emph{string} with \emph{length} denoted by $|S|$ is the concatenation of \emph{characters} $S[1]S[2]\cdots S[|S|]$ of $\Sigma$. Since it will be useful to have a predecessor and a successor for every character of the string, we also define $S[0]=\$$ and $S[|S|+1]=\$$ with $\$\notin\Sigma$. The \emph{substring} $S[i..j]$ with $0\leq i\leq j\leq |S|+1$ is the concatenation $S[i]S[i+1]\cdots S[j]$. For $i>j$ the substring $S[i..j]$ is defined to be the empty string with length $0$. A \emph{prefix} is a substring of the form $S[1..j]$ and a \emph{suffix} is a substring of the form $S[i..|S|]$.
    
    The string $S$ is \emph{lexicographically strictly smaller/larger} than the string $S'$ if $S$ is lexicographically smaller/larger than $S'$ and there is a mismatch $S[m]\neq S'[m]$.
    
    A \emph{maximal pair} of $S$ is a triple $(n,m,l)\in \mathbb{N}^3$ with 
    $l\geq 1$ such that $S[n..n+l-1]$ is equal to $S[m..m+l-1]$ and this property can not be extended to any side. More formally: 
    \begin{itemize}
        \item $\forall i\in \mathbb{N} \textup{ with } 0\leq i< l: S[n+i] = S[m+i]$ but
        \item $S[n-1] \neq S[m-1]$ and
        \item $S[n+l] \neq S[m+l]$.
    \end{itemize}
    With this notation, the string $S[n..n+l-1]=S[m..m+l-1]$ is the \emph{corresponding maximal repeat}.
    
    Since for a maximal pair $(n,m,l)$ the inequality $S[n-1] \neq S[m-1]$ holds, the indices $n$ and $m$ cannot be equal. Also, by construction, $S[n..n+l-1]$ and $S[m..m+l-1]$ are contained in $S$ and $S[n..n+l]$ and $S[m..m+l]$ are contained in $S\$$.
    
    Two maximal pairs $(n,m,l)$ and $(n',m',l')$ are \emph{copies of each other} if the two strings $S[n-1..n+l]$ and $S[m-1..m+l]$ are equal to the two strings $S[n'-1..n'+l']$ and $S[m'-1..m'+l']$. In particular, the two maximal pairs $(n,m,l)$ and $(m,n,l)$ are always copies of each other. However, it is not sufficient for two maximal pairs to have identical corresponding maximal repeats in order to be copies of each other.
    
    If two maximal pairs are not copies of each other, they are \emph{substantially different}.
    
    A \emph{period} of a string $S$ is an integer $\Delta$ such that all characters in $S$ with distance $\Delta$ are equal.
    
    Let $S[l..r]$ be a positioned $\Delta$-periodic substring with $|S[l..r]|\geq \Delta$. The \emph{maximal $\Delta$-periodic extension} of this occurrence is $(l',r')$ such that 
    \begin{itemize}
        \item $l'\leq l\leq r \leq r'$,
        \item $S[l'..r']$ is $\Delta$-periodic,
        \item $S[l'-1..r']$ is not $\Delta$-periodic and
        \item $S[l'..r'+1]$ is not $\Delta$-periodic.
    \end{itemize}
    With this notation, the pair $(l'-1,r'+1)$ is the \emph{padded maximal $\Delta$-periodic extension}.
    
    If $\Delta$ is the minimal period length of $S[l..r]$, we will omit the $\Delta$ and simply write \emph{maximal periodic extension} and \emph{padded maximal periodic extension}.
    
    Similar to maximal pairs, two padded maximal periodic extensions $(l-1,r+1)$ and $(l'-1,r'+1)$ are \emph{copies of each other} if the corresponding strings are equal. If the two padded maximal periodic extensions are not copies of each other, they are \emph{substantially different}.
    
    The (self-referential) \emph{LZ77-decomposition} of a string $S$ is a factorization $S=F_1 F_2 \dots F_n$ in LZ77-factors, such that for all $i\in \{1,2,\dots,n\}$ either
    \begin{itemize}
        \item $F_i$ is a character which does not occur in $F_1 F_2 \dots F_{i-1}$ or
        \item $F_i$ is a the longest possible prefix of $S[|F_1 F_2 \dots F_{i-1}|+1..|S|]$ which occurs at least twice in $F_1 F_2 \dots F_i$.
    \end{itemize}
    
    Let $\pi_i\in \{0,1,2,\dots, |S|\}$ be given by the lexicographic order of the acyclic permutations $S[{\pi_i}+1..|S|+1]S[1..\pi_i]$ of $S\$$. The \emph{Burrows-Wheeler transform} defined in \cite{BWT} is given by the last characters of those strings, and, since $S[0]=\$=S[|S|+1]$ hold by definition, these characters are given by $S[\pi_i]$.
    
    \section{Non-Highly-Periodic Maximal Pairs}
    
    The main goal of this section is to prove that in a string $S$ with $z$ LZ77-factors and without $q$-th powers the number of substantially different maximal pairs whose corresponding maximal repeats are not at least a sixth power is bounded from above by $41(\log_2 |S|)(z+1)(z+2)$. Along the way, it will also be shown that the CDAWG has at most $18q(1+log_q |S|)(z+2)^2$ arcs.
    
    In Theorem 8 of \cite{Pape-Lange}, I counted the number of maximal pairs around the boundaries of LZ77-factors which neither begin nor end with a power of a given exponent:
    
    \begin{theorem}[Theorem 8 of \cite{Pape-Lange}]\label{thm:specialUpperBound} 
        Let $S$ be a string.
        Let $F_1 F_2\dots F_z F_{z+1}=S\$$ be the LZ77-decomposition of $S\$$.
        Let $s_1, s_2,\dots , s_z, s_{z+1}$ be the starting indices of the LZ77-factors in $S\$$.
        Let $q\in \mathbb{N}_{\geq 2}$ and $i,j\in \{1,2,\dots, z, z+1\}$ be natural numbers.\\
        Then the number of different maximal pairs $(n_k, m_k, l_k)$ such that for all $k$
        \begin{itemize}
            \item the substring $S[n_k..s_i-1]$ is not a fractional power with exponent greater than or equal to $q$,
            \item the substring $S[s_i..n_k+l_k-1]$ is not a fractional power with exponent greater than or equal to $q$,
            \item the starting index $s_i$ is contained in the interval $[n_k, n_k+l_k]$,
            \item the starting index $s_{i+1}$ is not contained in the interval $[n_k, n_k+l_k]$ and
            \item the starting index $s_j$ is contained in the interval $[m_k, m_k+l_k]$
        \end{itemize}
        is bounded from above by $18q\cdot\lceil \log_{q}(|F_1 F_2\dots F_i|)\rceil$
    \end{theorem}
    
    This can be slightly simplified by ignoring the underlying LZ77-structure which is not used in the proof:
    
    \begin{corollary}
        Let $S$ be a string. Let $q\in \mathbb{N}_{\geq 2}$ be a natural number and $i, j$ be indices of two characters in $S\$$.
        
        Then there are at most $18q\cdot\lceil \log_{q}(|S\$|)\rceil$ different maximal pairs $(n_k, m_k, l_k)$ such that for all~$k$
        \begin{itemize}
            \item the substring $S[n_k..i-1]$ is not a fractional power with exponent greater than or equal to $q$,
            \item the substring $S[i..n_k+l_k-1]$ is not a fractional power with exponent greater than or equal to $q$,
            \item the index $i$ is contained in the interval $[n_k, n_k+l_k]$ and
            \item the index $j$ is contained in the interval $[m_k, m_k+l_k]$.
        \end{itemize}
    \end{corollary}
    
    Following the proof of Theorem 8 in \cite{Pape-Lange}, the interval $S[i..n_k+l_k-1]$ naturally splits into $S[n_k..i-1]$ and $S[i..n_k+l_k-1]$ and we can even require that the longer part(s) is/are not a high power(s):
    
    \begin{lemma} \label{lem:nonhighly}
        Let $S$ be a string. Let $q\in \mathbb{N}_{\geq 2}$ be a natural number and $i, j$ be indices of two characters in $S\$$.
        
        Then there are at most $18q (1 + \log_{q} |S|)$ different maximal pairs $(n_k, m_k, l_k)$ such that for all~$k$
        \begin{itemize}
            \item 
            \begin{itemize}
                \item if $|S[n_k..i-1]|\geq |S[i..n_k+l_k-1]|$, then $S[n_k..i-1]$ is not a fractional power with exponent greater than or equal to $q$,
                \item if $|S[n_k..i-1]|< |S[i..n_k+l_k-1]|$, then $S[i..n_k+l_k-1]$ is not a fractional power with exponent greater than or equal to $q$,
            \end{itemize}
            \item the index $i$ is contained in the interval $[n_k, n_k+l_k]$ and
            \item the index $j$ is contained in the interval $[m_k, m_k+l_k]$.
        \end{itemize}
    \end{lemma}
    
    As proven in Lemma 4 of \cite{Pape-Lange}, each maximal pair has a copy such that both double-sided extensions of the corresponding maximal repeats cross LZ77-boundaries. Also, each maximal pair introduces at most two new right extensions of maximal repeats. Therefore, we can deduce a bound similar to Theorem 1 of \cite{Pape-Lange} for the right extensions of maximal repeats and the arcs of the CDAWG:
    
    \begin{theorem}
        Let $S$ be a string. Let $z$ be the number of LZ77-factors of $S$. Let $q$ be a natural number such that $S$ does not contain $q$-th powers.
        
        Then the number of right extensions of maximal repeats in $S$ is bounded from above by $18q(1+log_q |S|)(z+2)^2 - (z+1)$. Also, the CDAWG of $S$ has at most $18q(1+log_q |S|)(z+2)^2$ arcs.
    \end{theorem}
    \begin{proof}
        Summing up over the first indices $i\leq j$ of the $z+1$ LZ77-factors of $S\$$ yields that there are at most
        \[\sum_{i=1}^{z+1}\sum_{j=i}^{z+1} 18q(1+\log_q |S|) 
        = 9q(1+log_q |S|)(z+1)(z+2) \leq 9q(1+log_q |S|)(z+2)^2 - (z+1)\]
        substantially different maximal pairs. Hence, there are at most $18q(1+log_q |S|)(z+2)^2 - (z+1)$ different right extensions of maximal repeats.
        
        Since the number of right extensions of (non-empty) maximal repeats is equal to the number of arcs in the CDAWG which start at internal nodes and since there are at most $|\Sigma\cup\{\$\}|\leq z+1$ arcs starting at the root, there are at most $18q(1+log_q |S|)(z+2)^2$ arcs in the CDAWG.
    \end{proof}
    
    Additionally, there might be maximal pairs, in which the longer part(s) is/are high power(s) but the corresponding periodicity does not extend to the whole maximal repeat. In order to find a good upper bound for those maximal pairs, we need an additional lemma to limit the number of possible period lengths of prefixes and suffixes with high powers.
    
    \begin{lemma}\label{lem:twoCubes}
        Let $S$ be a string.
        Let further $P_1$, $P_2$ be two substrings of $S$ such that
        \begin{itemize}
            \item either $P_1$ and $P_2$ are both prefixes of $S$ or $P_1$ and $P_2$ are both suffixes of $S$,
            \item the length of $P_2$ fulfills the inequality $|P_1| \leq |P_2|\leq 2|P_1|$ and
            \item both $P_1$ and $P_2$ are fractional powers with exponent greater than or equal to $3$.
        \end{itemize}
        Then $P_1$ and $P_2$ have the same minimal period length.
    \end{lemma}
    \begin{proof}
        Without loss of generality assume that $P_1$ and $P_2$ are both prefixes of $S$. Let $\Delta_1$ and $\Delta_2$ be the minimal period lengths of $P_1$ and $P_2$, respectively.
        
        On the one hand, $P_1$ is a substring of $P_2$. Therefore $P_1$ is $\Delta_2$-periodic. This implies $\Delta_1\leq \Delta_2$.
        
        On the other hand, the inequalities $\Delta_1 \leq \frac{1}{3} |P_1|$ and $\Delta_2\leq \frac{1}{3} |P_2|\leq \frac{2}{3} |P_1|$ hold. Therefore, for all indices $i$ with $1\leq i \leq \Delta_2$, the inequality $i+\Delta_1\leq |P_1|$ holds. Since $P_1$ is $\Delta_1$-periodic, the equation 
        \[S[\left((i) \mod \Delta_2\right) +1] = S[\left((i+\Delta_1) \mod \Delta_2\right) +1]\]
        holds as well. This implies that $P_2$ is $\Delta_1$-periodic. This, however, implies $\Delta_2\leq \Delta_1$ and thereby concludes the proof.
    \end{proof}
    
    \begin{theorem} \label{thm:nonhighly-nonextended}
        Let $S$ be a string. Let $i,j$ be indices of two characters in $S\$$.
        
        Then there are at most $12\log_{2} |S|$ different maximal pairs $(n_k, m_k, l_k)$ such that for all $k$
        \begin{itemize}
            \item 
            \begin{itemize}
                \item if $|S[n_k..i-1]|\geq |S[i..n_k+l_k-1]|$, then $S[n_k..i-1]$ is a fractional power with exponent greater than or equal to $3$ and period length $\Delta$,
                \item if $|S[n_k..i-1]|< |S[i..n_k+l_k-1]|$, then $S[i..n_k+l_k-1]$ is a fractional power with exponent greater than or equal to $3$ and period length $\Delta$,
            \end{itemize}
            \item the substring $S[n_k..n_k+l_k-1]$ is not $\Delta$ periodic,
            \item the index $i$ is contained in the interval $[n_k, n_k+l_k]$ and
            \item the index $j$ is contained in the interval $[m_k, m_k+l_k]$.
        \end{itemize}
    \end{theorem}
    \begin{proof} By contradiction:
        
        It is sufficient to prove that there are at most $6\log_{2} |S|$ different maximal pairs with the restrictions given by the prerequisites which fulfill $|S[n_k..i-1]|\geq |S[i..n_k+l_k-1]|$.
        By symmetry, the maximal pairs which fulfill the inequality $|S[n_k..i-1]|<|S[i..n_k+l_k-1]|$ can be bounded with an identical argument.
        
        Assume there are at least $\left\lfloor6\log_{2}(|S|)\right\rfloor+1$ different maximal pairs $(n_k, m_k, l_k)$ with $|S[n_k..i-1]|\geq |S[i..n_k+l_k-1]|$ and the restrictions given by the prerequisites.
        
        Since for all maximal pairs $1 \leq n_k$ holds, the inequality $i-n_k \leq |S\$|-1$ holds as well. On the other hand, since $S[n_k..i-1]$ is a fractional power with exponent greater than or equal to $3$, this substring has to contain at least three characters. Therefore, the inequality $3\leq i-n_k$ holds.
        
        Taking the logarithm yields
        \[1 < \log_{2}(3) \leq \log_{2}(i-n_k) \leq \log_{2}(|S\$|-1) \leq \lceil \log_{2}(|S|)\rceil\textup{.}\]
        For each maximal pair, the number $\log_{2}(i-n_k)$ lies in at least one of the $\lceil \log_{2}(|S|)\rceil - 1$ intervals $[h,h+1]$ with $1\leq h < \lceil \log_{2}(|S|)\rceil$.
        
        Using $\lceil \log_{2}(|S|)\rceil - 1 \leq \left\lfloor\log_{2}(|S|)\right\rfloor$, the pigeonhole principle now yields that there has to be a natural number $L'$ such that
        \[\left\lceil\frac{\left\lfloor6\log_{2}(|S|)\right\rfloor+1}{\left\lfloor\log_{2}(|S|)\right\rfloor}\right\rceil = 7\]
        of these maximal pairs have a starting index with $L' \leq \log_{2}(i-n_k) \leq 1+L'$.
        
        Therefore, for $L = 2^{L'}$ this gives a natural number $L$ such that $L\leq i-n_k\leq 2L$ holds for each of these $7$ maximal pairs.
        
        Since the index $i$ is contained in the interval $[n_k, n_k+l_k]$ and $|S[n_k..i-1]|\geq |S[i..n_k+l_k-1]|$ holds, the index $i$ is also contained in the interval $[n_k+\frac{l_k}{2}, n_k+l_k]$. Hence, the inequalities $n_k+\frac{l_k}{2} \leq i$ and thereby $\frac{l_k}{2} \leq i - n_k \leq 2L$ hold. Therefore, the length $l_k$ is at most $4L$.
        
        Since the index $j$ is contained in the interval $[m_k, m_k+l_k]$, this implies that the inequality $m_k\geq j-l_k \geq j-4L$ holds. On the other hand $m_k\leq j$ so the $m_k$ are in an interval of length~$4L$.
        
        Using the pigeonhole principle again, there are 
        \[\left\lceil\frac{7}{6}\right\rceil = 2\]
        of these maximal pairs $(n_a, m_a, l_a)$, $(n_b, m_b, l_b)$ such that the distance $|m_a-m_b|$ is at most~$\frac{2}{3}L$.
        
        According to Lemma \ref{lem:twoCubes}, both maximal pairs have the same minimal period length. Hence, the corresponding maximal repeats are of the form $p_a P^3 s_a r_a$ and $p_b P^3 s_b r_b$ where $p_a P^3$ and $p_b P^3$ are the $|P|$-periodic parts left of $i$, the substrings $s_a$ and $s_b$ are the maximal $|P|$-periodic extensions of $p_a P^3$ and $p_b P^3$ to the right and $r_a$ and $r_b$ are the remaining characters of the maximal repeats.
        
        Since the two $|P|$-periodic strings $p_a P^3 s_a$ and $p_b P^3 s_b$ starting at $n_a$ and $n_b$ overlap at least by $3|P|$ and since $s_a$ and $s_b$ are the maximal $|P|$-periodic extensions of $p_a P^3$ and $p_b P^3$ respectively, this implies that $s_a = s_b$. Therefore, the maximal repeats are of the form $p_a P^3 s r_a$ and $p_b P^3 s r_b$.
        
        Since $|m_a-m_b|\leq \frac{2}{3}L$ holds, we can show that the overlap of the $|P|$-periodic strings $p_a P^3 s$ and $p_b P^3 s$ starting at the indices $m_a$ and $m_b$ have at least an overlap of length $|P|$:
        
        The strings $p_a P^3 s$ and $p_b P^3 s$ have at least the length $3|P|$. Therefore, if $P\geq \frac{L}{3}$ holds, the overlap is at least $3|P|-\frac{2}{3}L\geq |P|$.
        
        The strings $p_a P^3 s$ and $p_b P^3 s$ also have at least the length $L$. Therefore, if $P\leq \frac{L}{3}$ holds, the overlap is at least $L-\frac{2}{3}L= \frac{L}{3} \geq |P|$.
        
        In either case, the overlap is at least as long as $P$.
        
        Therefore, the union of the occurrences of $p_a P^3 s$ and $p_b P^3 s$ starting at $m_a$ and $m_b$ is $|P|$-periodic. This implies that the strings $p_a P^3 s$ and $p_b P^3 s$ starting at $m_a$ and $m_b$ end with the same character.
        
        If the length of $p_a P^3 s$ and $p_b P^3 s$ is different, this implies that both occurrences of the smaller string starting at the indices $n_a$ and $m_a$ or at the indices $n_b$ and $m_b$ are preceded by the same character which is given by the $|P|$-periodic extension to the left. This, however, implies that either $(n_a, m_a, l_a)$ or $(n_b, m_b, l_b)$ is not a maximal pair.
        
        If, on the other hand, the length of $p_a P^3 s$ and $p_b P^3 s$ is equal, the starting indices $n_a$ and $n_b$ are equal and the starting indices $m_a$ and $m_b$ are equal as well. This, however, is only possible if either $(n_a, m_a, l_a)$ or $(n_b, m_b, l_b)$ is not a maximal pair or if both maximal pairs are identical.
        
        Since both cases contradict the assumption, the assumption is wrong and the theorem is therefore true.
    \end{proof}
    
    \begin{corollary}
        Let $S$ be a string. Let $q\in \mathbb{N}_{\geq 3}$ be an even natural number and $i, j$ be indices of two characters in $S\$$.
        
        Then there are at most $12(1+3\frac{q}{\log_2 q})\log_2 |S|$ different maximal pairs $(n_k, m_k, l_k)$ such that for all $k$
        \begin{itemize}
            \item the corresponding maximal repeat $S[n_k..n_k+l_k-1]$ is not a fractional power with exponent greater than or equal to $2q$,
            \item the index $i$ is contained in the interval $[n_k, n_k+l_k]$ and
            \item the index $j$ is contained in the interval $[m_k, m_k+l_k]$.
        \end{itemize}
        
        Let $z$ be the number of LZ77-factors of $S$. 
        
        Then there are at most $12(1+3\frac{q}{\log_2 q})(\log_2 |S|)(z+1)(z+2)$ different double-sided extensions of maximal repeats that are not a fractional power with exponent greater than or equal to $2q$.
    \end{corollary}
    \begin{proof}
        Without loss of generality, the inequality $q\leq |S|$ holds, since higher values for $q$ do not increase the number of permitted maximal pairs but do increase the upper bound.
        
        If a maximal repeat $[n_k, n_k+l_k]$ is not a fractional power with exponent greater than or equal to $2q$, then either the longer of the parts $S[n_k..i-1]$ and $S[i..n_k+l_k-1]$ is
        \begin{itemize}
            \item not a fractional power with exponent greater than or equal to $q$ or
            \item a fractional power with exponent greater than or equal to $q\geq3$ of which the periodicity does not extend to the whole maximal repeat $[n_k, n_k+l_k]$.
        \end{itemize}
        
        Therefore, the number of different maximal pairs which fulfill the prerequisites can be bound by Lemma \ref{lem:nonhighly} and Theorem \ref{thm:nonhighly-nonextended} and there are at most
        \[18q(1+\log_{q}|S|)+12(\log_{2}|S|)<36q(\log_{q}|S|)+12(\log_{2}|S|)=12(1+3\frac{q}{\log_2 q})\log_2 |S|\]
        of those maximal pairs.
        
        Summing up over the first indices $i\leq j$ of the $z+1$ LZ77-factors of $S\$$ yields that there are at most
        \[\sum_{i=1}^{z+1}\sum_{j=i}^{z+1} 12(1+3\frac{q}{\log_2 q})\log_2 |S| = 6(1+3\frac{q}{\log_2 q})\log_2 |S|(z+1)(z+2)\]
        substantially different maximal pairs, whose corresponding maximal repeat is not a fractional power with exponent greater than or equal to $2q$. 
        
        Hence, there are at most $12(1+3\frac{q}{\log_2 q})(\log_2 |S|)(z+1)(z+2)$ different double-sided extensions of maximal repeats that are not a fractional power with exponent greater than or equal to $2q$.
    \end{proof}
    
    For $q=3$ this proves that the number of substantially different maximal pairs whose corresponding maximal repeats is not at least a sixth power is bounded from above by $41(\log_2 |S|)(z+1)(z+2)$.
    
    \section{Highly-Periodic Maximal Pairs}
    
    The goal of this section is to prove that in a string $S$ with $z$ LZ77-factors the number of substantially different maximal pairs whose corresponding positioned maximal repeats are at least fourth powers of which at least one is not periodically extendable by more than one period length is bounded from above by $32(\log_2 |S|)(z+1)^2$.
    
    Both occurrences of those maximal pairs, including the corresponding maximal repeat as well as the preceding and succeeding characters, are inside of the two padded maximal periodic extensions of the corresponding positioned maximal repeats.
    
    Therefore, we will first count the number of substantially different padded maximal periodic extensions of fourth powers and the number of substantially different padded maximal periodic extensions of a given fourth power. Afterwards, we will show that each pair of padded maximal periodic extensions gives rise to at most $4$ maximal pairs whose corresponding positioned maximal repeats are at least fourth powers of which at least one is not periodically extendable by more than period length.
    
    We will need the ``Three Squares Lemma'' of Crochemore and Rytter presented in \cite{CrochemoreTSL}.
    
    \begin{lemma}
        Let $u$, $v$ and $w$ be primitive and let $u^2$, $v^2$ and $w^2$ be prefixes/suffixes of $S$ with $|u|<|v|<|w|$.
        
        Then $|w|>|u|+|v|$ holds.
    \end{lemma}
    
    \begin{lemma}
        Let $S$ be a string and $i$ be an index of a character in $S\$$.
        
        Then there are at most $4\left\lfloor\log_{2} |S|\right\rfloor$ substantially different padded maximal periodic extensions $(l-1,r+1)$ of fourth powers such that $l-1<i\leq r+1$.
    \end{lemma}
    \begin{proof}
        In this proof we will only count the number of padded maximal periodic extensions $(l-1,r+1)$ of fourth powers such that at least half of the interval $[l,r]$ is smaller than $i$, i.e. $l+\frac{r-l+1}{2} \leq i$. The other case $l+\frac{r-l+1}{2} \geq i$ is symmetrical.
        
        Since $S[l..r]$ is at least a fourth power, the string $S[l..i-1]$ is at least a square. Therefore, two maximal periodic extensions $(l,r)$ and $(l',r')$ of fourth powers have an  overlap of least twice the smaller minimal period length. Therefore, if their minimal period lengths are equal, the padded maximal periodic extensions $(l-1,r+1)$ and $(l'-1,r'+1)$ copies of each other. Conversely, if $(l-1,r+1)$ and $(l'-1,r'+1)$ are substantially different, then they have different minimal period lengths as well.
        
        This implies that the number of substantially different padded maximal periodic extensions $(l-1,r+1)$ of fourth powers such that at least half of the interval $[l,r]$ is smaller than $i$ is bounded from above by the number of different primitive squares that are suffixes of $S[1..i-1]$.
        
        The three squares lemma implies that for three primitive squares which are suffixes of each other, the largest square is more than twice as long as the smallest square.
        
        Since the smallest square has at least two characters and the largest square has at most $|S|$ characters, there are at most $2\left\lfloor\log_{2} |S|\right\rfloor$ primitive squares which are suffixes of $S[1..i-1]$.
        
        Therefore, there are at most $2\left\lfloor\log_{2} |S|\right\rfloor$ padded maximal periodic extensions $(l-1,r+1)$ of fourth powers such that at least half of the interval $[l,r]$ is smaller than $i$, i.e. $l+\frac{r-l+1}{2} \leq i$.
        
        This implies that the number of padded maximal periodic extensions of fourth powers $(l-1,r+1)$ such that $l-1<i\leq r+1$ is bounded from above by $4\left\lfloor\log_{2} |S|\right\rfloor$.
    \end{proof}
    
    The proof also allows another useful conclusion.
    
    \begin{corollary}
        Let $S$ be a string and $i$ be an index of a character in $S\$$. Furthermore, let $P$ be a substring of $S$ which is at least a fourth power.
        
        Then there are at most $2$ substantially different padded maximal periodic extensions $(l-1,r+1)$ of cyclic permutations of $P$ such that $l-1<i\leq r+1$.
    \end{corollary}
    
    Combining the previous corollary with the lemma before gives rise to an upper bound of the pairs of corresponding maximal periodic extensions.
    
    \begin{lemma}
        Each pair of padded maximal periodic extensions of fourth powers which are up to cyclic rotation identical gives rise to at most $4$ maximal pairs whose corresponding positioned maximal repeats are at least fourth powers of which at least one is not periodically extendable by more than period length.
    \end{lemma}
    \begin{proof}
        Each maximal pair has to be a prefix of the one padded maximal periodic extension and a suffix of the other padded maximal periodic extension, otherwise both corresponding positioned maximal repeats would be preceded or succeeded by the same character. There are two choices of which padded maximal periodic extension the corresponding positioned maximal repeat is a prefix.
        
        For a fixed choice, the length of the maximal repeat is fixed, up to a multiple of the period length. Therefore there are only two possible lengths such that at least one of the positioned maximal repeat is not periodically extendable by more than period length.
    \end{proof}
    
    This leads to the wanted upper bound:
    
    \begin{corollary}
        Let $S$ be a string. Let $z$ be the number of LZ77-factors in an LZ77-decomposition of $S$.
        
        Then there are at most $8(z+1)^2\log_{2} |S|$ substantially different pairs of padded maximal periodic extensions of fourth powers which are up to cyclic rotation identical.
        
        Also, there are at most $32(z+1)^2\log_{2} |S|$ substantially different maximal pairs whose corresponding positioned maximal repeats are at least fourth powers of which at least one is not periodically extendable by more than period length.
    \end{corollary}
    
    \section{RLBWT and Maximal Pairs}
    
    The goal of this section is to prove that the runs of the RLBWT of a string $S$ with $z$ LZ77-factors correspond to a subset of the maximal pairs, whose size can be bound from above by $73(\log_2 |S|)(z+2)^2$.
    
    Since we are interested in the number of runs, it is useful to observe the indices $i$ where a new run starts. These are exactly the index $1$ and the indices $i$ with $S[\pi_{i-1}]\neq S[\pi_i]$.
    
    Since $\$$ occurs exactly once in $S\$$, the strings $S[{\pi_{i-1}}+1..|S|+1]$ and $S[{\pi_i}+1..|S|+1]$ have a mismatch. Also, since $S[{\pi_{i-1}}+1..|S|+1]S[1..\pi_{i-1}]$ is lexicographically smaller than $S[{\pi_i}+1..|S|+1]S[1..\pi_i]$, the string $S[{\pi_{i-1}}+1..|S|+1]$ is lexicographically strictly smaller than $S[{\pi_i}+1..|S|+1]$.
    
    Let $m$ be the index of the first mismatch of these two strings. With this notation, the strings $S[{\pi_{i-1}}+1..{\pi_{i-1}}+m-1]$ and $S[{\pi_{i}}+1..{\pi_{i}}+m-1]$ are equal and their predecessors as well as their successors are different. Therefore, if $m>0$, they form a maximal pair. If $m=0$, then $S[{\pi_{i-1}}+1] < S[{\pi_{i}}+1]$. This, however can only occur $|\Sigma|$ times.
    
    On the other hand, since $S[{\pi_{i-1}}+1..{\pi_{i-1}}+m]$ is smaller than $S[{\pi_i}+1..|S|+1]S[1..\pi_i]$ and $S[{\pi_{i}}+1..{\pi_{i}}+m]$ is larger than $S[{\pi_{i-1}}+1..|S|+1]S[1..\pi_{i-1}]$, this maximal pair can only correspond to this pair $(\pi_{i-1},\pi_{i})$ of lexicographically neighbored acyclic permutations.
    
    \begin{remark}
        Belazzougui et al.\ show in Theorem 1 of \cite{Belazzougui:lower_bound} that the number of runs in the Burrows-Wheeler transform is even bounded in the number of right extensions of the maximal repeats. However, maximal pairs are easier to handle then right extensions of maximal repeats and we only lose a factor $\Sigma$ in the worst-case by not using the right extensions.
    \end{remark}
    
    However, while the number of maximal repeats and thereby the number of nodes in the CDAWG can be $\Theta (qz^3)$ for a suitable set of strings, the Burrows-Wheeler transform does not suffer from high powers as the CDAWG does:
    
    \begin{lemma} \label{lem:non-extendability}
        Let $S$ be a string and let $i$ be an index at which a new run in the Burrows-Wheeler transform starts. 
        Let $P$ be the maximal repeat $S[{\pi_{i}}+1..{\pi_{i}}+m-1]$ of the associated maximal pair and
        let $p$ be a prefix of $P$ such that there is a rational number $q$ with $P=p^q$.
        
        Then the maximal $|p|$-periodic extension of either
        \begin{itemize}
            \item $S[{\pi_{i-1}}+1..{\pi_{i-1}}+m-1]$ or
            \item $S[{\pi_{i}}+1..{\pi_{i}}+m-1]$ or
            \item both occurrences
        \end{itemize}
        contains less than $|p|+1$ additional characters.
    \end{lemma}
    \begin{proof}
        Assume that the maximal $|p|$-periodic extensions of both occurrences $S[{\pi_{i-1}}+1..{\pi_{i-1}}+m-1]$ and $S[{\pi_{i}}+1..{\pi_{i}}+m-1]$ contain at least $|p|+1$ additional characters. In this proof, we will show that under this assumption, there is a cyclic permutation $S[w+1..|S|+1]S[1..w]$ of $S\$$ which is lexicographically between $S[{\pi_{i-1}}+1..|S|+1]S[1..\pi_{i-1}]$ and $S[{\pi_{i}}+1..|S|+1]S[1..\pi_{i}]$.
        
        If the maximal $|p|$-periodic extension of $S[{\pi_{i-1}}+1..{\pi_{i-1}}+m-1]$ extends this occurrence to the left, the equation $S[{\pi_{i-1}}]=S[{\pi_{i-1}}+|p|]$ and thereby \[S[{\pi_{i}}]\neq S[{\pi_{i-1}}]=S[{\pi_{i-1}}+|p|] = S[{\pi_{i}}+|p|]\] holds. Therefore, the maximal $|p|$-periodic extension of $S[{\pi_{i}}+1..{\pi_{i}}+m-1]$ does not extends this string to the left. This implies that at most one of the two maximal $|p|$-periodic extensions of the occurrences $S[{\pi_{i-1}}+1..{\pi_{i-1}}+m-1]$ and $S[{\pi_{i}}+1..{\pi_{i}}+m-1]$ does extend the occurrence to the left.
        
        Similarly, at most one of those occurrences is extended to the right by the maximal $|p|$-periodic extension.
        
        Since, by assumption, both occurrences are $|p|$-periodically extendable, exactly one occurrence has to be $|p|$-periodically extendable to the left and exactly one occurrence has to be $|p|$-periodically extendable to the right. By symmetry we can assume without loss of generality that $S[{\pi_{i-1}}+1..{\pi_{i-1}}+m-1]$ is $|p|$-periodically extendable to the left and that $S[{\pi_{i}}+1..{\pi_{i}}+m-1]$ is $|p|$-periodically extendable to the right.
        
        Hence, $S[{\pi_{i-1}}-|p|..{\pi_{i-1}}+m-1]$ and $S[{\pi_{i}}+1..{\pi_{i}}+m+|p|]$ are $|p|$-periodic. Also, by definition of the Burrows-Wheeler transform, the inequality $S[{\pi_{i-1}}+m] < S[{\pi_{i}}+m]$ holds.
        
        Combining the periodicity with this inequality yields \[S[{\pi_{i-1}}+1..{\pi_{i-1}}+m-1]=S[{\pi_{i-1}}+1-|p|..{\pi_{i-1}}+m-1-|p|]\] 
        and 
        \[S[{\pi_{i-1}}+m] < S[{\pi_{i}}+m] = S[{\pi_{i}}+m-|p|] =  S[{\pi_{i-1}}+m-|p|]\]
        which imply
        \[S[{\pi_{i-1}}+1..|S|+1]S[1..\pi_{i-1}] < S[{\pi_{i-1}}+1-|p|..|S|+1]S[1..\pi_{i-1}-|p|]\textup{.}\]
        
        Similarly, we get 
        \[S[{\pi_{i-1}}+1-|p|..{\pi_{i-1}}+m-1]=S[{\pi_{i}}+1..{\pi_{i}}+m-1+|p|]\]
        and 
        \[S[{\pi_{i-1}}+m]< S[{\pi_{i}}+m] = S[{\pi_{i}}+m+|p|]\]
        which imply
        \[S[{\pi_{i-1}}+1-|p|..|S|+1]S[1..\pi_{i-1}-|p|] < S[{\pi_i}+1..|S|+1]S[1..\pi_i]\textup{.}\]
        
        Since $S[{\pi_{i-1}}+1-|p|..|S|+1]S[1..\pi_{i-1}-|p|]$ is lexicographically between the cyclic permutations $S[{\pi_{i-1}}+1..|S|+1]S[1..\pi_{i-1}]$ and $S[{\pi_i}+1..|S|+1]S[1..\pi_i]$, these two strings are not neighbors with regard to the Burrows-Wheeler transform. This contradicts the assumption and thereby concludes the proof.
    \end{proof}
    
    Therefore, the positioned maximal repeats of the associated maximal pairs corresponding to the RLBWT are either not highly-periodic or, if they are highly-periodic, the period cannot be extended by more than a period length. This implies the following corollary and thereby leads to another proof of the Burrows-Wheeler conjecture:
    
    \begin{corollary}
        Let $S$ be a string with $z$ LZ77-factors.
        
        Then, there are at most $73(\log_2 |S|)(z+2)^2$ runs in the RLBWT.
    \end{corollary}
    
    \section{Conclusion}
    
    This paper proved that of the potentially $\mathcal{O}(qz^3)$ substantially different maximal pairs in a string, it is sufficient to understand a subset containing at most $73(\log_2 |S|)(z+2)^2$ maximal pairs.
    
    It seems therefore likely that it is possible to merge the nodes of the CDAWG which correspond maximal repeats of the same base and get a new data structure which is almost as universal and intuitive as the CDAWG but does only has $\mathcal{O}((\log |S|)z^2)$ arcs.
    
    Also, the proofs presented in this paper do not use the underlying structure of the string. If the substrings of $S$ and the reversed string $S_{\operatorname{rev}}$ are also highly compressible and have less than $z'$ LZ77-factors each, it should be possible to prove that the number of runs in the RLBWT is bounded from above by $\mathcal{O}(z' z^2)$.
    
    Thereby, it might be possible to derive an upper bound for the runs in the RLBWT which is only dependent on the number of LZ77-factors.
    
    \section{Acknowledgements}
    
    Fabio Cunial suggested that my previous work might be extendable from counting maximal repeats to counting extensions of maximal repeats. He also pointed out that such a result would be more interesting since it is more closely linked to the size of the compacted directed acyclic word graph. Nicola Prezza noted that my previous work also resulted in a non-trivial upper bound for the number of runs in the run-length Burrows-Wheeler transform and that a more careful investigation of the extensions of maximal repeats might result in a better bound for the Burrows-Wheeler conjecture which was unsolved at that time. I also thank Djamal Belazzougui for notifying me of the ``resolution of the Burrows-Wheeler conjecture'' by Kempa and Kociumaka.
    
    \bibliography{maximalRepeats}
    
\end{document}